\documentclass{article}
\usepackage[utf8]{inputenc}
\usepackage{amsthm, amssymb,amscd,amsmath,amsfonts}
\usepackage{hyperref}
\hypersetup{colorlinks=true, citecolor=black, urlcolor=blue,linkcolor=blue}
\usepackage{color}
\usepackage{xcolor}
\usepackage{todonotes}
\usepackage{IEEEtrantools}
\usepackage{mathtools}
\usepackage{chngcntr}
\usepackage{soul} 

\usepackage{svg}

\usepackage[shortlabels]{enumitem}

\usepackage{float}

\usepackage{subcaption}
\usepackage{multirow}
\usepackage[round, authoryear]{natbib} 
\usepackage{bbm}
\usepackage{lineno} 

\newtheorem{theorem}{Theorem}
\newtheorem{lemma}[theorem]{Lemma}

\newtheorem{proposition}[theorem]{Proposition}
\newtheorem{remark}[theorem]{Remark}
\newtheorem{corollary}[theorem]{Corollary}

\newcommand\N{{\mathbb N}}
\newcommand\R{{\mathbb R}}

\renewcommand\P{{\mathbb P}}

\newcommand\E{{\mathbb E}}

\allowdisplaybreaks

\numberwithin{equation}{subsection}
\makeatletter
\@addtoreset{equation}{section}
\makeatother

\usepackage{authblk}
\title{Small-time approximation of the transition density for diffusions with singularities. Application to the Wright-Fisher model.}
\author[1,2]{T.Roa \thanks{tania.roa@uai.cl}}
\author[3]{M. Fariello \thanks{fariello@fing.edu.uy}}
\author[3,4]{G. Martínez}
\author[3,5]{J.R. Le\'on \thanks{rlramos@fing.edu.uy}}
\affil[1]{Universidad Adolfo Ibáñez. Faculty of Engineering \& Sciences. Viña del Mar, Chile}
\affil[2]{Data Observatory Foundation, ANID Technology Center No. DO210001, Eliodoro Yáñez 2990, 7510277, Providencia, Santiago, Chile}
\affil[3]{Universidad de La Rep\'ublica. IMERL. Uruguay}
\affil[4]{McGill University, Canada\thanks{gerardo.martinez@mail.mcgill.ca}}
\affil[5]{Escuela de Matemática UCV. Venezuela}
\date{\today}                     
\begin{document}

\maketitle

\begin{abstract}
The Wright-Fisher (W-F) diffusion model serves as a foundational framework for interpreting population evolution through allele frequency dynamics over time. Despite the known transition probability between consecutive generations, an exact analytical expression for the transition density at arbitrary time intervals remains elusive. Commonly utilized distributions such as Gaussian or Beta inadequately address the fixation issue at extreme allele frequencies (0 or 1), particularly for short periods. In this study, we introduce two alternative parametric functions, namely the Asymptotic Expansion (AE) and the Gaussian approximation (GaussA), derived through probabilistic methodologies, aiming to better approximate this density.

The AE function provides a suitable density for allele frequency distributions, encompassing extreme values within the interval [0,1]. Additionally, we outline the range of validity for the GaussA approximation. While our primary focus is on W-F diffusion, we demonstrate how our findings extend to other diffusion models featuring singularities. Through simulations of allele frequencies under a W-F process and employing a recently developed adaptive density estimation method, we conduct a comparative analysis to assess the fit of the proposed densities against the Beta and Gaussian distributions.\\
\textbf{Key words:} 
diffusions with singularities,
Wright-Fisher model, 
transition density. 

\textbf{Classification:} 
\end{abstract}


\section{Introduction}


The Wright-Fisher (W-F) model, initially proposed by \cite{Wright1931}, offers a robust mathematical framework for studying the dynamics of allele frequencies within a population over successive generations. 

To begin, we introduce the discrete  W-F model alongside its diffusion approximation. The discrete W-F model, as delineated by \cite{Wright1931}, operates within a population of constant size $N$ and two allelic types $\{a,A\}$. It is represented as a Markov chain $X_n$ with state space $\{0,1,2,\ldots,N\}$, where $X_n$ denotes the count of allele $a$ at generation $n$. Transition probabilities within this model are defined by a conditional Bernoulli law:
$$\P(X_{n+1}=j\,|X_n=i)=\binom{N}{j}\left(\frac iN\right)^{j}\left(1-\frac iN\right)^{N-j},\quad j=0,\ldots,N.$$

Scaling the number of generations relative to the population size  permits approximate the model with a diffusion process.  Let $X_n^{N}$ represent the Markov chain, emphasizing the population size, with $\frac1N X^N_0\to x_0$ as $N\to\infty$. Introducing the continuous-time process $X_N(t)=\frac1NX^N_{[Nt]}, t\ge0,$ yields  $X_N(\cdot)\to X(\cdot)$ as $N$ tends to infinity, being this convergence  in distribution, as demonstrated in the book \cite{Ethier1986}. The limiting process is a Markov diffusion with state space  $[0,1]$, governed by the following Stochastic Differential Equation (SDE):
\begin{eqnarray}\label{difusion}dX(t)&=&\sqrt{X(t)(1-X(t))}dW(t)\\
\nonumber X(0)&=&x_0
\end{eqnarray}

Here, $W(t)$ denotes a standard Wiener process.

The infinitesimal generator of this diffusion process is defined as:
$$\mathcal L(f)(x)=\frac12(x(1-x))\frac{d^2f}{dx^2}(x),$$ where $f$ is a function twice continuously differentiable and defined on  $(0,1)$ with appropriate boundary conditions.

Let us introduce the main focus of our work. Since no simple analytical expression is available for the distribution of allele frequencies in a given generation, the application of the discrete W-F model in inference has relied on approximate results. 
An overview of these results can be found in \cite{Tataru2016}, which compares several approximations of the W-F distribution used in statistical inference under different assumptions.

The Gaussian distribution is commonly employed ( \citet{Bonhomme2010, Coop2010, Gautier2010, Pickrell2012, Lacerda2014, Terhorst2015}), albeit with limitations due to the discrete nature of the model and the appearance of fixation or loss events at the extremes of the frequency spectrum. Most of the authors cite \citet{Nicholson2002} to justify the use of this distribution, but no proof is given in that article. This approximation is not always appropriate due to (i) the discrete nature of the model, (ii) the range of frequencies, which is restricted to the interval $[0,1]$, and (iii) the appearance of atoms at 0 and 1 in the allele frequency, corresponding to fixation and loss of an allele, respectively. As the number of generations increases, given a constant population size, the atoms have an increasing probability of remaining fixed at 0 or 1. One possible solution, explored in the work of \cite{Tataru2015}, is to use a Beta distribution, Beta with spikes at the extremes of the interval $[0,1]$, or a truncated Gaussian distribution. The most promising proposal in terms of goodness of fit of the allele frequency distribution is to consider the diffusion model \eqref{difusion}, but since its transition density  does not have a closed analytical expression, it is not possible to perform statistical inference. These drawbacks make finding analytical expressions and/or valid ranges of the Gaussian distribution very important for  statistical inference.



In this paper we propose a more general framework that allows obtaining a finite definite density function on the interval $[0,1]$. Our methodology, inspired by seminal works in the field, provides an asymptotic expansion of the probability density function solution for the W-F diffusion model, particularly effective when $\frac n{2N}$ is small. To find such a density, we generalize to a singular diffusion the technique developed in the work of \cite{DacunhaCastelle1986} and \cite{Voronka1975}. Other works that complement the scope of the cited work of Voronka and Keller are those of \cite{antonelli1978}, \cite{tier1978}, and \cite{tier1981}, where the Ohta–Kimura model with two-locus di-allelic systems with linkage are studied, for example.

As well as \cite{Voronka1975}, our method provides an asymptotic expansion of the solution of the probability density function of the forward Kolmogorov equation for the W-F diffusion model that is valid when $t=\frac n{2N}$ is small (where $n$ is the number of generations and $N$ is the constant population size). It is important to point out that our method is probabilistic in essence, instead of asymptotic methods in partial differential equations as is the case of the aforementioned work. 

While our primary focus remains on the W-F diffusion, we present a generalized solution applicable to diffusions with singularities parameterized by 
two parameters, $a$ and $b$. When these parameters are equal to $\frac{1}{2}$, we obtain W-F diffusion. Thus, we show how to develop a general procedure to find the transition density of a diffusion that is state dependent and singular.

To evaluate the efficacy of our proposed density and compare it with existing methods, we conduct simulations of discrete W-F trajectories.  For a fixed time $t, \, t \in [0,1]$, we use an adaptive estimation method to obtain a continuous empirical density from the simulated data (see, e.g., \citet{Bertin2011}). Having a continuous density from the simulated data allows us to better evaluate the fit of different continuous densities to the data using the Hellinger distance and the $\mathbb L^{2}$ norm. 

Our work is organized into various sections. In Section \ref{section: preliminaries}, we introduce our notation, the underlying model that we are investigating, the process of transforming  state dependent diffusion, significant findings in discovering transition densities of diffusion processes, and the lemma that provides the transition density. We provide a comprehensive proof of this lemma in Section \ref{proof-lemma}. For Section \ref{section: approximation}, we present significant outcomes associated with the W-F model, as well as corresponding proofs and special scenarios, such as neutral, with mutation and selection. Proofs for these scenarios can be found in Section \ref{apend}. Additionally, in Section \ref{section: model evaluation}, we assess the suitability of the defined models and two other models commonly used: the Beta and  Gaussian distribution.

\section{Preliminaries}\label{section: preliminaries}

In this section we begin by defining a stochastic differential equation (SDE) that has a drift component, $\mu(\cdot)$, and a local variance, $\sigma(\cdot)$, depending on $X(t)$. We then define the adaptation of Lemma 2, presented by \cite{DacunhaCastelle1986} paper, so that the transition density can be studied. In this lemma, both the drift term and the local variance term require certain regularity of these functions. Using the Lamperti transform (see, for example, \cite{moller2010}), we study a transformed process and the conditions for the drift term of this process to fulfill the conditions of the proposed lemma.

We will find in subsection \ref{principal} (proposition \ref{proposition asumption}) an exact formula for the transition density of the following diffusion model
\begin{equation} \label{sigma-gen}
dX(t) = X(t)^{a} (1 - X(t))^{b} dW(t); \qquad a, b, X(t) \in (0,1),
\end{equation}
where $W$ is a standard Wiener process. The W-F diffusion is a particular case of this model,  when  $a=b=\frac{1}{2}$. 

As  local variance term $\sigma(x)=x^a(1-x)^b$ of this SDE is not constant, so it is necessary to transform it.

%

\subsection{Local variance transformation for a SDE}
Through the Lamperti transformation we can convert a SDE, with drift and local variance term dependent on $t$, into a SDE with a constant local variance term. \\

Let us consider the following stochastic differential equation, 
\begin{equation} \label{general}
dX(t) =b (X(t)) dt + \sigma (X(t)) dW(t),
\end{equation}
where its drift $b(\cdot)$ component is s smooth function, and the local varainace $\sigma(\cdot)>0$ is twice differentiable in $[0,1]$, the continuity of this function implies that it is lower bounded by a constant $c>0$, and $W$ ia a standard Wiener process.\\

Let $Y(t) = F(X(t))$ be a transformation of $X(t)$, where $F(\cdot)$ is  a twice differentiable function; Itô lemma yields
\begin{equation*}
    dY(t)= \left[ F'(X(t))b(X(t)) +\frac{1}{2} F''(X(t))\sigma^2(X(t))\right] dt+F'(X(t))\sigma(X(t)) \, dW(t).
\end{equation*}

Defining  $F(\cdot)$ as 
\begin{align} \label{eq: definition of g}
& F(x) = \int_{0}^{x} \frac{1}{\sigma(u)} du, \quad \text{ we have,}\nonumber \\
& F'(x) =  \frac{1}{\sigma(x)}, \text{and,} \\
& F''(x) =  -\frac{\sigma'(x)}{\sigma^{2}(x)}. \nonumber
\end{align}
Then the result for the transformed process is 

\begin{equation}\label{eq: transformed process 1}
    dY(t)= \left [ F'(X(t))b(X(t))+\frac{1}{2} F''(X(t)) \sigma^2(X(t)) \right ]dt + dW(t). 
\end{equation}
The process $Y(t)$ is now the solution of a SDE with a constant local variance term equal to one. With the additional hypothesis that $F(\cdot)$ has an inverse, i.e $x =  F^{-1}(y)$, we can re-write the equation \eqref{eq: transformed process 1} as
\begin{multline}\label{eq: transformed process 2}
    dY(t)= \left [ F'(F^{-1}(Y(t)))b(F^{-1}(Y(t)))\right.\\\left .+\frac{1}{2} F''(F^{-1}(Y(t))) \sigma^2(F^{-1}(Y(t))) \right]dt + dW(t). 
\end{multline}

Now, we have a SDE, resulting from the Lamperti transformation, with constant local variance equal to one and drift term equal to
$$\mu(y)=F'(F^{-1}(y))b(F^{-1}(y))+\frac{1}{2} F''(F^{-1}(y)) \sigma^2(F^{-1}(y)).$$

In the particular case when $b=0$ we obtain
\begin{equation} \label{sde-transformed}
dY(t) = - \dfrac{1}{2}\sigma'(F^{-1}(Y(t))) dt + dW(t).
\end{equation}

\subsection{Transition density for a diffusion with constant local variance term}
The above transformation allows us to work, without loss of generality, with a diffusion that has  a constant term of local variance. Now based on Lemma 2, of \cite{DacunhaCastelle1986}, we present the following lemma that allows obtaining a exact formula for the transition density.

\begin{lemma}\label{lemma: transition density}[Transition density]
Let us consider the following stochastic differential equation
\begin{equation} \label{general-cte}
dX(t) = \mu(X(t)) dt + dW(t)
\end{equation}
where $W(t)$ is the standard Wiener process and $\mu(\cdot)$ is a smooth drift function and verifies 
\begin{equation}\label{eq: moment}
    \mu^2(x)+\mu'(x)=O(|x|^2)\mbox{ when } |x|\to\infty.
\end{equation}

Let $p_t$ be the transition density of the process defined in \eqref{general-cte}, and $\{B(t)\}_{t \geq 0}$ be a standard Brownian bridge and let us denote $M(t) = \int_0^t \mu(s) \, ds$. Then, $p_t$ can be written as
\begin{multline}\label{eq: lemma transition density}
    p_t(x,y) =  \frac{1}{\sqrt{2\pi t}} \exp \left [- \frac{(x-y)^2}{2t} + (M(y) - M(x))\right ] \\
     \times \, \E \exp \left [-\frac t2\int_0^1 \nu(\sqrt tB(s)+(1-s)x+sy)ds  \right ],
\end{multline}
where $\nu(x) = (\mu^2(x) + \mu'(x))$.
\end{lemma}

A detailed proof of Lemma \ref{lemma: transition density} can be found in the following section. 

Furthermore, using the above lemma we can obtain an useful approximation for the transition density when $t\to0$.\\


\begin{lemma} \label{expec-td}
For $t \to 0$ and $y \neq x$, we have
\begin{equation*}
\lim_{t \to 0} \frac{1}{t} \E \exp \left [-\frac t2\int_0^1 \nu(\sqrt tB(s)+(1-s)x+sy)ds  \right ] = -\frac{1}{2(y-x)}\int_x^y \nu(u) \, du
\end{equation*}

\end{lemma}

\begin{proof}
Using the approximation $\exp(f(t))-1 \approx f(t)$ when $f(t) \to 0$, we get that
\begin{multline*}
     \lim_{t \to 0 }\frac{1}{t} \left [\exp \left (-\frac t2\int_0^1 \nu(\sqrt tB(s)+(1-s)x+sy)ds  \right ) - 1\right ]\\ \approx \lim_{t \to 0} -\frac{1}{2}\int_0^1 \nu(\sqrt tB(s)+(1-s)x+sy)ds\\
      =  -\frac{1}{2}\int_0^1 \nu((1-s)x+sy) \, ds 
      =   -\frac{1}{2(y-x)}\int_x^y \nu(u) \, du.
\end{multline*}
Then
\begin{multline}\label{eq: approximation exponential}
    \frac{1}{t}\E \left \{\exp \left [-\frac t2\int_0^1 \nu(\sqrt tB(s)+(1-s)x+sy)ds  \right ] - 1 \right \}\\ \underset{t \to 0}{\longrightarrow}-\frac{1}{2(y-x)}\int_x^y \nu(u) \, du
\end{multline}
\end{proof}

\begin{corollary} \label{AE1}
 [Asymptotic expansion] The transition density, $p_{t}$, can be written as follows
\begin{multline}\label{eq: asymptotic expansion}
 p_t(x,y)=\frac{1}{\sqrt{2\pi t}} \exp \left [- \frac{(x-y)^2}{2t} + (M(y) - M(x))\right ]\\\times \big(1-\frac t{2(y-x)}\int_x^y \nu(u) \, du+o(t)\big).
\end{multline}
\end{corollary}

Above, we have used the Landau's symbols $o(\cdot)$ and $O(\cdot)$, we say the $d(t)=o(t)$ or $O(t))$, if $\lim_{t\to0}\frac{d(t)}t=0$ or a constant, respectively. \\

The equation defined in \eqref{eq: asymptotic expansion} gives us an approximation to the transition density that is valid when $t$ approaches 0. For this reason, we will refer to this approximation as the \textbf{Asymptotic Expansion (AE)}. \\

Given that we are going to work with a transformed process $Y(t)$, it is necessary to establish a link between the results for the original process, $X(t)$ and the transformed one. The following lemma is presented in order to state the relation between their two transition densities.

\begin{lemma}\label{lemma: density transformed process}
Let $X(t)$ be a real-valued Markov process with transition density $p_t^{X}$ and $F$  an increasing and differentiable function. Let us define the process $Y(t) = F(X(t))$ with transition density $p_t^Y$. Then, for any real numbers $x_0 <x$, the transition density of $X(t)$ can be written as
\begin{equation}\label{eq: density transformed process}
    p_t^X(x_0, x) = p_t^Y(F(x_0), F(x))F'(x).
\end{equation}
\begin{proof}
Let $c_{1}, c_{2}$ be real numbers such that $c_{1} < c_{2}$. Then,
\begin{IEEEeqnarray*}{rCl}
    \P^{x_0}(X(t)\in [c_{1}, c_{2}])=\P^{y_0}(Y(t)\in[F(c_{1}),F(c_{2})])&=&\int_{F(c_{1})}^{F(c_{2})}p^Y_t(y_0,y)dy \\
    &=&\int_{c_{1}}^{c_{2}} p^Y_t(F(x_0),F(x))F'(x)dx,
\end{IEEEeqnarray*}
where $y_0=F(x_0)$ and $y=F(x)$.
By the definition of transition density, we conclude that \eqref{eq: density transformed process} holds.
\end{proof}
\end{lemma}
This lemma implies that if we have an exact or approximate analytical expression for the transition density of $Y(t)$, we also have one for $X(t)$. 
%

\section{Proof for transition density lemma \ref{lemma: transition density}} \label{proof-lemma}

To study the transition density defined in \eqref{eq: lemma transition density}, we further analyze the terms involved in it. Starting with the associated semigroup and how it is related to the drift function, continuing with the associated drift function and finally and analysis of the function $\sigma(\cdot)$ considered. \\

For ease of notation let us write the transition density of $Y(t)$ as 
\begin{multline}\label{eq: density of y}
    p_t^Y(y_0,y) =  q_t(y_0,y) \exp \left (M(y_{0}) - M(y)\right ) \\
     \times \, \E \exp \left [-\frac t2\int_0^1 \nu(\sqrt tB(s)+(1-s)x+sy)ds  \right ],
\end{multline}
where \begin{equation*}
    q_t(y_0,y)=\frac{1}{\sqrt{2\pi t}}\exp \left [-\frac{(y_0-y)^2}{2t} \right ] \text{ and }M(y) = \int_0^y - \dfrac{1}{2}\sigma'(F^{-1}(u)) du \, dy.
\end{equation*}

By equation \eqref{eq: density transformed process}, we obtain
\begin{multline}\label{eq: density of y 2}
    p_t^X(x_0,x) =  q_t(F(x_0)),F(x)) \exp \left [M(F(x_0)) - M(F(x))\right ] F'(x)\\
     \times \, \E^{F(x_0)} \exp \left [-\frac t2\int_0^1 \nu(\sqrt tB(s)+(1-s)x+sy)ds  \right ].
\end{multline}

\subsection{Semigroup associated}
The semigroup associated to the Markov process $X$ defined in \eqref{general-cte} is given by
\begin{equation}\label{eq: markov semigroup}
    P_tf(x)=\E^x[f(X(t))]=\int_Ip_t(x,y)f(y)dy,
\end{equation}
where $f(\cdot)$ is a smooth function with compact support $I\subset [0,1]$ and $\E^x$ is the expectation conditioning on $X = x$. The semigroup in \eqref{eq: markov semigroup} can be written using the Cameron-Martin-Girsanov formula (see \cite{Oksendal2000} pp. 123) (this formula holds under the condition $\mu^2(x)=O(|x|^2)$ when $|x|\to\infty$) as 
\begin{IEEEeqnarray}{rCl}
 P_tf(x) & = & \E^x\left[e^{\int_0^t  \mu(W(s))dW(s)-\frac{1}{2}\int_0^t \mu^2(W(s))ds}f(W(t))\right]\label{semigroup2} 
 \end{IEEEeqnarray}
 
Which, when it is applied to \eqref{sde-transformed}, gives
\begin{equation} \label{girsanov1}
\E^{y_0}[f(Y(t))]=\E^{y_0}[e^{-\frac12\int_0^t\sigma'(F^{-1}(W(s)))dW(s)-\frac14\int_0^t(\sigma'(F^{-1}(W(s)))^2ds} f(W(t))].
\end{equation}
 
Now, we analyze the drift function involved.
 
\subsection{Drift function}
In our case, the drift function is defined by $\mu (z) =-\frac{1}{2}\sigma'(F^{-1}(z))$. Considering that $M(z) = \int_{0}^{z} \mu(s) ds$ and $F(\cdot)$ is an increasing function, we have
\begin{align*}
M(z) &= - \dfrac{1}{2} \int_{F^{-1}(0)}^{{F^{-1}(z)}} \dfrac{\sigma'(x)}{\sigma(x)} dx = - \dfrac{1}{2} \left( \log \sigma (F^{-1} (z)) - \log \sigma (F^{-1} (0)) \right), \\
M'(z) &=-\frac12\sigma'(F^{-1}(z)),\\
M''(z) &=-\frac12\sigma''(F^{-1}(z))\sigma(F^{-1}(z)),
\end{align*}

Considering these results and by means of Itô lemma,
\begin{align*}
& M(W(t))-M(W(0))=-\frac12\int_0^t\sigma'(F^{-1}(W(s)))dW(s) - \frac12\int_0^t\sigma''(F^{-1}(W(s)))\sigma(F^{-1}(W(s)))ds 
\end{align*}
Implying that, 
\begin{align*}
& M(W(t))-M(W(0))+\frac12\int_0^t\sigma''(F^{-1}(W(s)))\sigma(F^{-1}(W(s)))ds=-\frac12\int_0^t\sigma'(F^{-1}(W(s)))dW(s) 
\end{align*}

Replacing this result in \eqref{girsanov1}, it results a formula without terms involving stochastic integrals
\begin{equation}\label{girsanov2}
\E^{y_0}[f(Y(t))]=e^{-M(W(0))}\E^{y_0}[e^{\frac12\int_0^t\sigma''(F^{-1}(W(s)))\sigma(F^{-1}(W(s)))ds-\frac14\int_0^t(\sigma'(F^{-1}(W(t)))^2ds}e^{M(W(t))}f(W(t))].
\end{equation}

\begin{remark} It is important to point out that all the results obtained above can be obtained without difficulty from the paper of \cite{DacunhaCastelle1986}.
\end{remark}

\subsection{Analysis of $\sigma(\cdot)$}\label{principal}
This subsection aims to extend the precedent results to more general local variance terms. First let us assume that $\sigma: [0,1] \to \mathbb{R}^{+}$ is concave and let us define
$$
V(x)=\frac12|\sigma''(F^{-1}(x))|\sigma(F^{-1}(x)))+\frac14(\sigma'(F^{-1}(x))^2\ge0.
$$

The above result enables us to write,
\begin{multline*}\E^{y_0}[f(Y(t))]=\int_{F(0)}^{F(1)} p^Y_t(y_0,y)f(y)dy\\=
\int_{F(0)}^{F(1)} \E[e^{-\frac12\int_0^tV(y_0+W(s))ds}\,|W(t)=y-y_0]q_t(y_0-y)e^{M(y)-M(y_0)}f(y)dy
\\=\int_{F(0)}^{F(1)} \E[e^{-\frac12\int_0^tV(y_0+(W(s)-\frac stW(t))+\frac stW(t))ds}\,|W(t)=y-y_0]q_t(y_0-y)e^{M(y)-M(y_0)}f(y)dy
\\=\int_{F(0)}^{F(1)} \E[e^{-\frac12\int_0^tV((1-\frac st)y_0+\frac st y+(W(s)-\frac stW(t))ds}]q_t(y_0-y)e^{M(y)-M(y_0)}f(y)dy
\\=\int_{F(0)}^{F(1)}\E[e^{-\frac t2\int_0^1V((1-u)y_0+u y+\sqrt t(W(u)-uW(1))du}]q_t(y_0-y)e^{M(y)-M(y_0)}f(y)dy.\end{multline*}

By duality, it results that for $F(0)\le y_0,y\le F(1)$, yields,

$$ p^Y_t(y_0,y)=\E[e^{-\frac t2\int_0^1V((1-u)y_0+u y+\sqrt t(W(u)-uW(1))du}]q_t(y_0-y)e^{M(y)-M(y_0)}.$$

In addition, for  $y_0=F(x_0)$
$$\P^{x_0}(X(t)\in[a,b])=\P^{y_0}(F(X(t))\in[F(a),F(b)])=\int_{F(a)}^{F(b)}p^Y_t(y_0,y)dy=\int_{a}^{b}p^Y_t(F(x_0),F(x))\frac1{\sigma(x)}dx.$$

This relationship means that
\begin{equation}\label{transformacion}
p^X_t(x_0,x)=p^Y_t(F(x_0),F(x))\frac1{\sigma(x)}.
\end{equation}

Below we extend the class of functions $\sigma(\cdot)$ for which the above result holds, based on the following assumptions

\begin{itemize}
\item[\textbf{(A1)}] Let us consider that $\sigma:[0,1]\to\R^{+}$, is concave, and it tends to zero at $0$ and $1$, such that $\int_0^1\frac1{\sigma(u)}du<\infty,$. 

\item[\textbf{(A2)}] We can approximate $\sigma$ and its two derivatives by a sequence of functions with two continuous derivatives such that $\sigma''_n\le0$, $\sigma_n\to\sigma$, as well as its derivatives
\end{itemize}

These assumptions allow us to approximate the expression \eqref{girsanov2} by a similar expression written for the sequence instead of $\sigma$; then, by the dominated convergence theorem,we can exchange limit with  expectation. \\

To clarify this point, let $V_{n}$ be defined as $V$, but replacing $\sigma_n(\cdot)$ instead of $\sigma(\cdot)$, and let us denote $p^{Y_n}_t$ the transition density. The Kolmogorov differential equation, implies that $p^{Y_n}_t(y_0,y)\to p^Y_t(y_0,y)$. \\

Additionally, we have 
\begin{multline}\label{limite}\E[e^{-\frac t2\int_0^1V_n((1-u)y_0+u y+\sqrt t B_n(u) du}]q_t(y_0-y)e^{M_n(y)-M_n(y_0)}
\\ =\E[e^{-\frac t2\int_0^1V_n((1-u)y_0+u y+\sqrt t B(u)du}]q_t(y_0-y)e^{M_n(y)-M_n(y_0)}.\end{multline}

where, by equality in distribution, we have substituted, the sequence of Brownian bridges, $B_{n}(u)$, by a fixed Brownian bridge $B(u)$. \\

Considering that $V_n$ is a positive function, we use dominated convergence theorem to show that the term \eqref{limite} converges to
$$\E[e^{-\frac t2\int_0^1V((1-u)y_0+u y+\sqrt t B(u)du}]q_t(y_0-y)e^{M(y)-M(y_0)}.$$

Thus, we obtain
$$ 
p^Y_t(y_0,y)=E[e^{-\frac t2\int_0^1V((1-u)y_0+u y+\sqrt t B(u) du}]q_t(y_0-y)e^{M(y)-M(y_0)},
$$

and also the relationship \eqref{transformacion}, holds. \\

These computations allow us to give the following  result 
\begin{proposition} \label{proposition asumption}
Based on assumptions (A1) and (A2), we get
$$ 
p^Y_t(y_0,y)=E[e^{-\frac t2\int_0^1V((1-u)y_0+u y+\sqrt t B(u) du}]q_t(y_0-y)e^{M(y)-M(y_0)},
$$
\end{proposition}

\begin{remark}
We think this is a new result; we have not found a similar one in the literature.. The class of functions  defined by $\sigma_{a,b}(x)=x^a(1-x)^b$, for $0<a,b<1$ satisfies the conditions \textbf{(A1)} and \textbf{(A2)}. When $a=b=\frac12$ it corresponds to the W-F diffusion.
\end{remark}

\section{Wright - Fisher diffusion: Asymptotic Expansion and Gaussian approximation}\label{section: approximation}
In the following we concentrate in the main subject of this paper: the W-F diffusion ( $a=b=1/2$). 

A parallel development can be carried out for other values of the parameters, although we will not consider this problem in the present work.
\subsection{Analytical expression for the transition density}\label{section: AE}
The goal of this section will be to prove an exact analytical expression for the transition density and its asymptotic expansion for small $t$, of the W-F diffusion. When this diffusion is considered, it produces the following
\begin{align}
& dX(t) = \sqrt{X(t) (1 - X(t))} dW(t) \label{sde-WF} \\
& dY(t) = - \dfrac{1}{2} \cot (Y(t)) dt + dW(t) \label{transf-WF}
\end{align}

Let us recall that in this case the following function writes $\nu(y) =\mu^2(F^{-1}(y)) + \mu'(F^{-1}(y))= \frac{1}{4} \cot^{2}(y) + \frac{1}{2} \csc^{2}(y) = \frac{1}{2} + \frac{3}{4} \cot^{2}(y)$. The main result is 
\begin{theorem}\label{theo: density W-F}
The transition density of the W-F diffusion can be written as 
\begin{multline}(i) \qquad\label{eq: AE1}
p^X_t(x_0,x)=\frac{(x_0(1-x_0))^{\frac{1}{4}}}{(x(1-x))^{\frac{3}{4}}}\\\times\E \left[e^{-\frac t2\int_0^1\nu((1-u)F(x_0)+u F(x)+\sqrt tB(u))du}\right]q_t(F(x),F(x_0)),
\end{multline}

and its asymptotical expansion as
\begin{multline}
(ii) \qquad p^X_{t}(x_0,x) = \frac{(x_{0} (x - x_{0}))^{\frac14}}{(x(1-x))^{\frac34}} \exp \left [-\frac {1}{2t} (F(x) - F(x_{0}))^{2}  \right ] \\
\times\left[ 1 - \frac{t}{2}  \, \frac{1}{F(x) - F(x_{0})} \int_{F(x_{0})}^{F(x)} \nu (u) du + o(t) \right] \label{eq: AE2}
\end{multline}
for all $x_0, x \in (0,1)$.
\end{theorem}

\begin{proof}
Let us start by proving $(i)$. For ease of notation let us write the transition density of $Y(t)$, based on Proposition \ref{proposition asumption}, as 
\begin{multline}\label{eq: density of y}
    p_t^Y(y_0,y) =  q_t(y_0,y) \exp \left (M(y) - M(y_{0})\right ) \\
     \times \, \E \exp \left [-\frac t2\int_0^1 \nu(\sqrt tB(s)+(1-s)y_{0}+sy)ds  \right ],
\end{multline}
where \begin{equation*}
    q_t(y_0,y)=\frac{1}{\sqrt{2\pi t}}\exp \left [-\frac{(y_0-y)^2}{2t} \right ] \text{ and }M(y) = \int_0^y - \frac{1}{2}\cot(y) \, dy.
\end{equation*}
By equation \eqref{eq: density transformed process}, we obtain
\begin{multline}\label{eq: density of y 2}
    p_t^X(x_0,x) =  q_t(F(x_0)),F(x)) \exp \left [M(F(x_0)) - M(F(x))\right ] \frac{1}{\sqrt{x(1-x)}}\\
     \times \, \E \exp \left [-\frac t2\int_0^1 \nu(\sqrt tB(s)+(1-s)F(x_{0})+sF(x))ds  \right ].
\end{multline}
Let us further develop the terms in equation \eqref{eq: density of y}. For every pair of positive real numbers $y$ and $y_0$ such that $y > y_0$, we have that
\begin{IEEEeqnarray*}{rCl}
 M(y)-M(y_0)  &= &  -\frac{1}{2}\int_{y_0}^{y}\cot(u)\,du \\
  & = &  \frac{1}{2} \left (\log|\sin y|- \log|\sin y_0| \right ) \\
   & = & \frac{1}{4} \left[\log\sin^2(y_0)-\log\sin^2(y) \right] \\
   & = & \frac{1}{4}\left [\log \left (\sin^2 \left (\frac{y_0}{2} \right )\cos^2 \left (\frac {y_0}{2} \right ) \right ) -\log \left (\sin^2 \left (\frac {y}{2} \right )\cos^2 \left (\frac {y}{2} \right ) \right)\right ],
\end{IEEEeqnarray*}
where, in the last equation, it was used that $\sin(2\theta) = 2\sin(\theta)\cos(\theta)$ for all $\theta \in \R$. Then, using the definition of $F$ in \eqref{eq: definition of g} and setting $F(x_0)=y_0 $ and $F(x)=y$, we obtain
\begin{equation}\label{eq: inverse}
F^{-1}(y) = \sin ^2(\frac{y}{2}), \quad F'(F^{-1}(y))=\frac1{\sin y} \quad \text{and} \quad F''(F^{-1}(y))=-4\,\frac{\cot y}{\sin^2y}.
\end{equation}

Replacing this result, it yields
\begin{equation*}
    M(F(x))-M(F(x_0))=\frac{1}{4}\big(\log(x_0(1-x_0))-\log(x(1-x))\big)=\log \left (\frac{x_0(1-x_0)}{x(1-x)} \right )^{\frac{1}{4}},
\end{equation*}
and thus,
\begin{equation}\label{eq: exp M}
    \exp \left [M(F(x_0)) - M(F(x))\right ] = \left (\frac{x_0(1-x_0)}{x(1-x)} \right )^{\frac{1}{4}}.
\end{equation}

Plugging \eqref{eq: exp M} in equation \eqref{eq: density of y} we obtain
\begin{multline}\label{eq: density of x 1}
p^X_t(x_0,x)=\frac{(x_0(1-x_0))^{\frac{1}{4}}}{(x(1-x))^{\frac{3}{4}}}\\\times\E \left[e^{-\frac t2\int_0^1\nu((1-u)F(x_0)+u F(x)+\sqrt tB(u))du} \right]q_t(F(x),F(x_0)).
\end{multline}

To show $(ii)$, we proceed as \eqref{eq: asymptotic expansion}. Using Taylor expansion we get \eqref{eq: AE2}.
\end{proof}

\begin{remark}
The result in $(ii)$ is completely in accord to the one of \citet{Voronka1975} (formula (3.16)). Although we have worked with the function, under its notation, $a(x)=\sqrt{x(1-x)}$ and with a drift zero in the equation \eqref{eq: WF-selection}, the same procedure can be used for considering more general functions $a(\cdot)$ and $b(\cdot)$.
\end{remark}
\clearpage

\subsection{Representation of the exact analytical expression}
To gain some intuition, the following figures depict the behavior of the exact analytical expression derived in the equation $\eqref{eq: density of x 1}$. The expectation was calculated by Monte Carlo method, simulating 500 realizations of the Brownian bridge. 
\begin{figure}[h!]
    \begin{subfigure}[t]{\textwidth}
        \centering
        \includegraphics[scale=0.2]{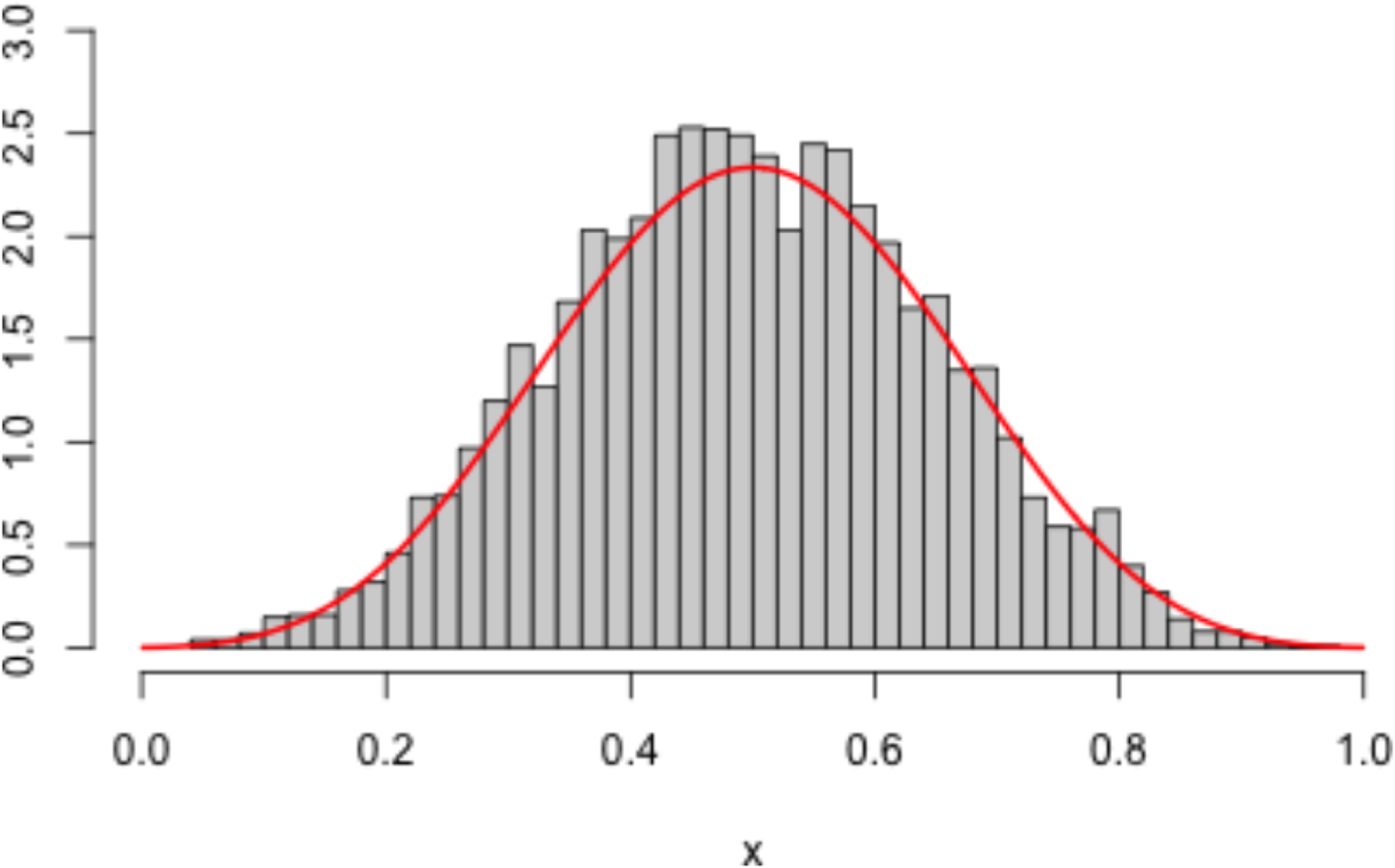}
        \caption{$t=0.1$}
    \end{subfigure}%
        \hfill
    \begin{subfigure}[t]{\textwidth}
        \centering
        \includegraphics[scale=0.2]{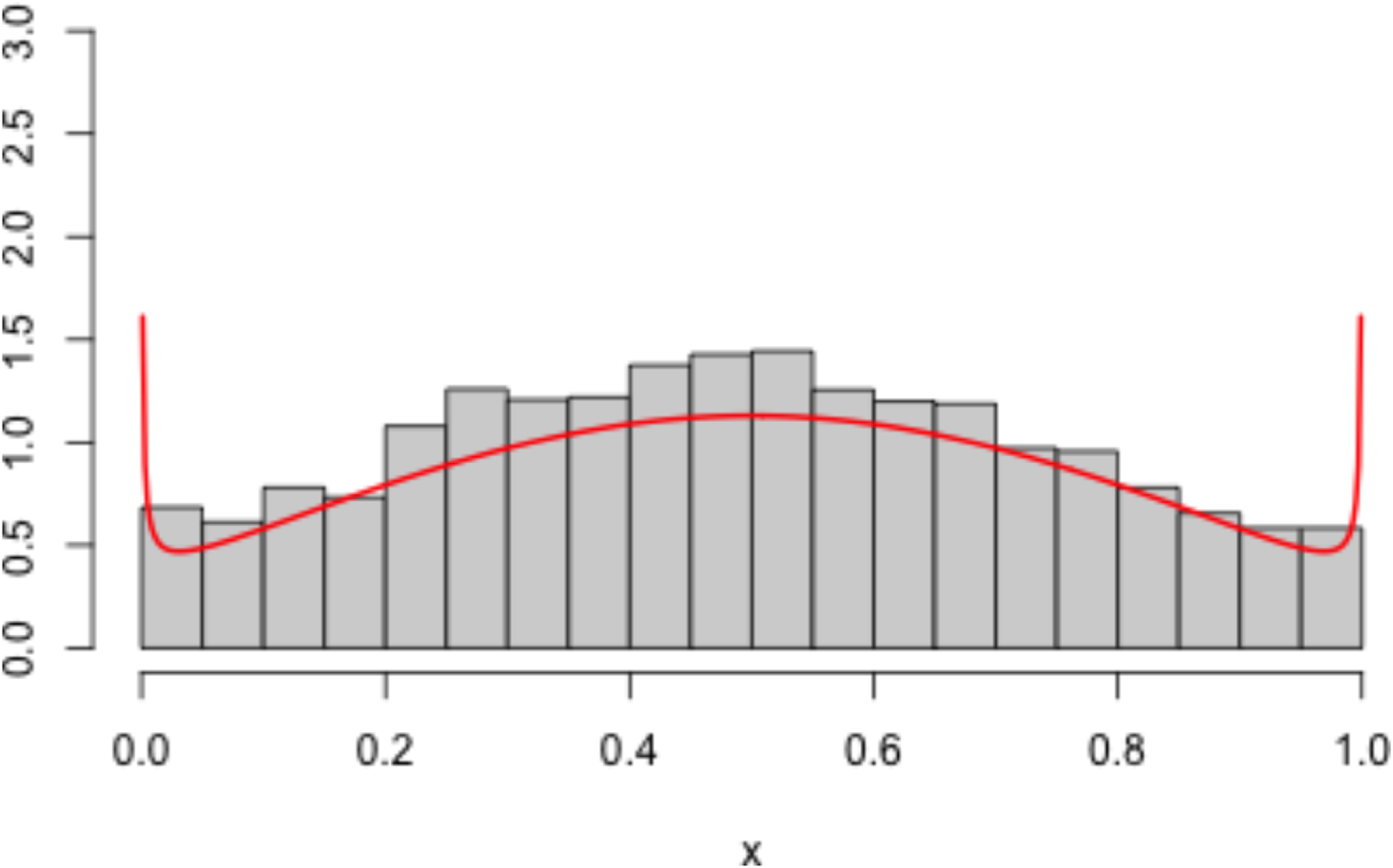}
        \caption{$t=0.3$}
    \end{subfigure}
    \hfill
\end{figure}
\clearpage
\begin{figure}[h!] \ContinuedFloat
    \begin{subfigure}[t]{\textwidth}
        \centering
        \includegraphics[scale=0.2]{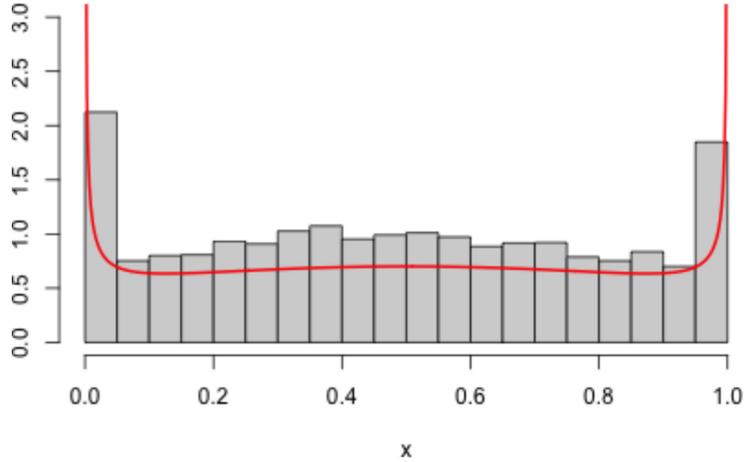}
        \caption{$t=0.5$}
    \end{subfigure}%
\caption{Representation of the exact analytical expression (red line), and simulated process (histogram)  for the transition density at different fixed times $t$, starting from $x_0$ .} \label{exact-analytical}
\end{figure}

For the three time points shown in Figure \ref{exact-analytical}, the new exact analytical formula (red line) we developed for the transition density fits the true distribution of the allele frequency for a given time and captures well the behavior at the extremes of the interval. 

In practice, this newly developed analytical expression enables us to precisely calculate the likelihood and to state the estimation problem of determining $\hat{x}_0$, i.e., the estimated initial allele frequency. However, to perform parametric inference, it is necessary to consider multiple populations. This involves finding the joint density and studying the relationship between the different populations generated. These subjects are not considered here and will be left for further research.

\subsection{Gaussian approximation}\label{section: gaussian approximation}
In several estimation problems linked to W-F model, it is common to use  the following Gaussian approximation
\begin{equation*}
    p^X_{t}(x_0,x)\approx\frac1{\sqrt{2\pi t}}\frac{1}{(x_0(1-x_0))^{\frac{1}{2}}}e^{-\frac {1}{2t}\frac{(x-x_0)^2}{(x_0(1-x_0))}}.
\end{equation*}

In this section we prove this result theoretically as an approximation of the transition density proved in Theorem \ref{theo: density W-F} and we study under which theoretical conditions it holds. However, the result wil be valid only over a range of values linking $t$ and $x$.

\begin{theorem}\label{theo: gaussian approximation}
The transition density in Theorem \ref{theo: density W-F} can be approximated by the Gaussian density
\begin{equation}
    \varphi_t^X(x_0, x) = \frac{1}{\sqrt{2\pi t}}\frac{1}{(x_0(1-x_0))^{\frac{1}{2}}}\exp \left (-\frac {1}{2t}\frac{(x-x_0)^2}{(x_0(1-x_0))} \right ).
\end{equation}
\begin{proof}
Let $p_t^X$ be the the AE of the W-F process from the equation \eqref{eq: AE1}, and $\sigma(x) = \sqrt{x(1-x)}$. Let us introduce the following  auxiliary function
\begin{equation*}
    \tilde p^X_{t}(x_0,x)=\frac1{\sqrt{2\pi t}} \frac{(x_0(1-x_0))^{\frac{1}{4}}}{(x(1-x))^{\frac{1}{4}}(x_0(1-x_0))^{\frac{1}{2}}} \exp \left [-\frac{1}{2t}\Big(\int_{x_0}^{x}\sigma^{-1}(u)du\Big)^2 \right],
\end{equation*}

recalling the definition of $F(x)$ in \eqref{eq: definition of g}. We can write
\begin{equation}\label{GA1}
    p^X_{t}(x_0,x)=\tilde p^X_{t}(x_0,x)+(p^X_{t}(x_0,x)-\tilde p^X_{t}(x_0,x)).
\end{equation}

Let us study the second term on the right hand side. Substituting $p^X_{t}$ by its formula, we have
\begin{multline*}
    |p^X_{t}(x_0,x)-\tilde p^X_{t}(x_0,x)|\\= \left |\frac1{\sqrt{2\pi t}} \frac{(x_0(1-x_0))^{\frac{1}{4}}}{(x(1-x))^{\frac{1}{4}}}\frac{[1-x_0+x][(x_0(1-x_0))^{\frac{1}{2}}+(x(1-x))^{\frac{1}{2}}]^{-1}}{(x_0(1-x_0))^{\frac{1}{2}}(x(1-x))^{\frac{1}{2}}} \right.\\
    \times \left. \exp \left [-\frac{1}{2t} \left (\int_{x_0}^{x}\sigma^{-1}(u)du \right)^2 \right ](x-x_0) \right |+O\left(t \right).
\end{multline*}

Plugging this into \eqref{GA1}, it yields 
\begin{multline*}
p^X_{t}(x_0, x) =\frac1{\sqrt{2\pi t}}\frac{(x_0(1-x_0))^{\frac{1}{4}}}{(x(1-x))^{\frac{1}{4}}}\frac{1}{(x_0(1-x_0))^{\frac{1}{2}}}\exp\left [-\frac{1}{2t} \left (\int_{x_0}^{x}\sigma^{-1}(u)du \right )^2 \right ]\\
\times\left(1+\frac{[1-x_0+x][(x_0(1-x_0))^{\frac{1}{2}}+(x(1-x))^{\frac{1}{2}}]^{-1}}{(x(1-x))^{\frac{1}{2}}}(x-x_0)+O\left(t \right) \right).
\end{multline*}

Since,
\begin{equation} \label{limit}
    \lim_{x\to x_0}\frac{[1-x_0+x][(x_0(1-x_0))^{\frac{1}{2}}+(x(1-x))^{\frac{1}{2}}]^{-1}}{(x(1-x))^{\frac{1}{2}}}=\frac1{2(x_0(1-x_0))}=O(1).
\end{equation}
we can then write
\begin{multline*}
p^X_{t}(x_0,x)=\frac1{\sqrt{2\pi}} \frac{(x_0(1-x_0))^{\frac{1}{4}}}{(x(1-x))^{\frac{1}{4}}}\frac{1}{(x_0(1-x_0))^{\frac{1}{2}}}\exp \left [-\frac{1}{2t}\Big(\int_{x_0}^{x}\sigma^{-1}(u)du\Big)^2 \right ]\\ \times (1+O(|x-x_0|)+O(t)).
\end{multline*}

The same type of expansion as above in a neighborhood of $x_0$ allows us to obtain
\begin{equation*}
    p^X_{t}(x_0,x)=\frac1{\sqrt{2\pi}}\frac{1}{(x_0(1-x_0))^{\frac{1}{2}}}e^{-\frac{1}{2t}(\int_{x_0}^{x}\sigma^{-1}(u)du)^2}(1+O(|x-x_0|)+O(t)).
\end{equation*}

Finally, to obtain the Gaussian approximation, we first need the following estimate
\begin{multline*}|e^{-\frac{1}{2t}(\int_{x_0}^{x}\sigma^{-1}(u)du)^2}-e^{-\frac{1}{2t}\frac{(x-x_0)^2}{(x_0(1-x_0))}}|\\
=e^{-\frac{1}{2t}\frac{(x-x_0)^2}{(x_0(1-x_0))}}|e^{-\frac{1}{2t}[(\int_{x_0}^{x}\frac1{\sqrt{(u(1-u))}})du)^2-(\int_{x_0}^x\frac1{\sqrt{(x_0(1-x_0)})}du)^2)]}-1|.\end{multline*}

To bound this expression, let us assume that
\begin{itemize}
    \item $\min\{\min\{1-x_0,1-x\};\min\{x_0,x\}\}>\delta$, and
    \item $\frac{1}{t}(x-x_0)^2$ is small enough.
\end{itemize}

Then, by using the Taylor expansion for the exponential function, we get
\begin{equation*}
    |e^{-\frac{1}{2t}[(\int_{x_0}^{x}\frac1{\sqrt{(u(1-u))}})du)^2-(\int_{x_0}^x\frac1{\sqrt{(x_0(1-x_0)})}du)^2)]}-1|\le\mathbf C\delta^{-2}\frac{1}{t}(x-x_0)^2.
\end{equation*}

Thus, finally all this yields
\begin{multline*}
p^X_{t}(x_0,x)=\frac1{\sqrt{2\pi t}}\frac{1}{(x_0(1-x_0))^{\frac{1}{2}}}e^{-\frac{1}{2t}\frac{(x-x_0)^2}{(x_0(1-x_0))}}\\\times (1+O(\frac{1}{t}(x-x_0)^2)+O(|x-x_0|)+O(t)).
\end{multline*}

Therefore, by choosing $|x-x_0|=O(t)$, we obtain the desired approximation. Indeed
\begin{equation*}
    p^X_{t}(x_0,x)=\frac1{\sqrt{2\pi t}}\frac{1}{(x_0(1-x_0))^{\frac{1}{2}}}e^{-\frac{1}{2t}\frac{(x-x_0)^2}{(x_0(1-x_0))}}(1+O(t)).
\end{equation*}
\end{proof}
\end{theorem}

\begin{remark}
Through the proof we should note that the Gaussian density holds while $t$ is small, \emph{i.e.} $n << 2N$ but also that $|x-x_0|=O(t)$. If $t$ is small, the last condition gives a range of application of the Gaussian approximation for statistical inference. Despite this drawback, when $|x-x_0|$ is large with respect to $t$ the presence in the product of the $q_t(x-x_0)$ term makes the contribution to the density small for this range of values.
\end{remark}

\subsection{Generalization: mutation and selection}
The W-F model can be generalized to include mutation and selection. Following the notation in \cite{Ewens2004}, the W-F diffusion with mutation and selection is the solution of the following SDE 
\begin{equation}\label{eq: WF-selection}
    \arraycolsep=1.4pt
    \left \{ \begin{array}{rcl}
    dX(t)&=& \mu(X(t)) \, dt + \sigma(X(t)) \, 
    dW(t)\\
    X(0)&=& x_0,
    \end{array}\right.
\end{equation}
where $\mu(x) =  \alpha x(1-x)(x+h(1-2x)) - \beta_1x + \beta_2(1-x)$, $\sigma(x) = \sqrt{x(1-x)}$, $h$ is heterozygosity, and $\alpha$, $\beta_1$ and $\beta_2$ are rescaled selection and mutation rates, respectively. \\

The result stated in Theorem (\ref{theo: density W-F}) can be extended to the more general W-F process with mutation and selection defined in the equation \eqref{eq: WF-selection}. 

\begin{proposition}\label{theo: density W-F mutation}
Let $\{X(t)\}_{t \geq 0}$ be the the W-F diffusion with selection and mutation. Under a purely mutational regime, i.e. with the selection coefficient $\alpha$ equal to 0, the transition density can be written as 
\begin{equation*}
    p^X_{t}(x_0,x)=\frac1{\sqrt{2\pi t}} \frac{(1-x_0)^{(\frac14-\beta_1)}x_0^{(\frac14-\beta_2)}}{(1-x)^{\frac34-\beta_1}x^{\frac34-\beta_2}}\exp\Big[{-\frac{1}{2t}(\int_{x_0}^x\sigma^{-1}(u)du)^2}\Big](1+O(t)).
\end{equation*}
for $x, x_0 \in (0,1)$. If the mutation coefficients are equal to $0$, then the transition density can be written as:
\begin{equation*}
    p^X_{t}(x_0,x)=\frac1{\sqrt{2\pi t}} \frac{x_0^{(\frac14-h)}}{x^{(\frac34-h)}}\frac{(1-x_0)^{(\frac14+\alpha-h)}}{(1-x)^{(\frac34+\alpha-h)}}\exp\Big[{-\frac{1}{2t}(\int_{x_0}^x\sigma^{-1}(u)du)^2}\Big](1+O(t)).
\end{equation*}
\end{proposition}

We leave its proof to the appendix \ref{OSC1} and \ref{OSC2}, respectively.

\section{Model evaluation through simulations: results and discussion}\label{section: model evaluation}

To evaluate the theoretical results proven previously, we first simulated the discrete W-F process to obtain the transition density after a certain number of generations for a fixed population size. Secondly, we fitted a continuous density and then compared this density with the ones proposed studied in section \ref{section: approximation}.

We simulated the evolution of the allele frequency in a population with constant size ($2N=1000$) for 500 generations, with initial allele frequencies ($x_0$) ranging from $x_0=0.1$ to $x_0 = 0.9$. The allele frequencies at generation $n$ were computed as the number of sampled alleles divided by $2N$. The alleles at generation $n$ were sampled from the previous one from a Binomial distribution with parameters $2N$ and $X_{n-1}/2N$. Rescaling the time as $t=n/2N$ we have the process $X_t$, where $t$ varies from 0 to 0.5, with steps of size 0.001. For each $x_0$ we simulated $100$ trajectories.   





We compared the empirical distribution obtained from the simulations to two commonly used parametric distributions (Gaussian and Beta), to the asymptotic expansion found in section \ref{section: AE} (AE), and to the Gaussian approximation found in section \ref{section: gaussian approximation} (GaussA). 

To further describe the two parametric models, let us recall that given $X_0 = x_0$ the allele frequency $X_n$ in the discrete W-F model has known expected value and variance given by
\begin{equation*}
        \E(X_t) = x_0 \quad \text{and} \quad  \mathrm{var}(X_t) = \left (1 - \left (1  - \frac{1}{2N}\right)  \right)^nx_0(1-x_0).
\end{equation*}
Proof of this statement can be found in \cite{Tataru2015}. With the re-scaling of time such as $t = n/2N$ we obtain
\begin{equation}\label{eq: expectation and variance WF}
        \E(X_t) = x_0 \quad \text{and} \quad  \mathrm{var}(X_t)  \approx e^{-t}x_0(1-x_0).
\end{equation}

The four proposed models are thus:

\begin{enumerate}[i.]
    \item A Beta distribution with parameters $\alpha_t$ and $\beta_t$   such that the probability density function of $X_t$ is
        \begin{equation*}
            p(x; \alpha_t, \beta_t) = \frac{x^{\alpha_t-1}(1-x)^{\beta_t-1}}{\mathrm{B}(\alpha_t, \beta_t)},
        \end{equation*}
        where $\mathrm{B}(\alpha, \beta)$ is the beta function. Following \cite{Tataru2015}, the parameters $\alpha_t$ and $\beta_t$ will be written in terms of the analytical expressions for the  expected value and variance of $X_t$ 
        \begin{IEEEeqnarray}{rCl}
        \alpha_t &=& \left (\frac{\E(X_t)(1-\E(X_t))}{\mathrm{var}(X_t)} \right) \E(X_t) \\
        \beta_t & = & \left (\frac{\E(X_t)(1-\E(X_t))}{\mathrm{var}(X_t)} \right) (1-\E(X_t)).
        \end{IEEEeqnarray}
        
    \item A Gaussian distribution with mean and variance coming from the equation \eqref{eq: expectation and variance WF}.
    \item The AE derived in Section \ref{section: AE} without the term $O(t)$, i.e.
    \begin{equation*}
    p(x;x_0)=\frac1{\sqrt{2\pi t}} \frac{(x_0(1-x_0))^{\frac{1}{4}}}{(x(1-x))^{\frac34}}\exp \left [-\frac {1}{2t} \left (\int_{x_0}^{x}\sigma^{-1}(u)du \right )^2 \right ]
    \end{equation*}
    where $\sigma(x) = \sqrt{x(1-x)}$.
    \item The Gaussian approximation found in theorem \ref{theo: gaussian approximation}. To ease the notation we will refer to this approximation as GaussA.
\end{enumerate}

The main difference between the densities in (ii) and (iv) are the approximations of the variance. While for the density in (ii) we used $e^{-t}p_0(1-p_0)$, in (iv) we used its approximation $tp_0(1-p_0)$. These two expressions take almost the same values while $0<t<0.2$, but then the variance in (iv) grows faster than the variance in (ii).
In addition, the \textit{densities} in (iii) and (iv) are not true densities as they do not integrate 1. Hence they need to be corrected by a normalization constant. The method used to calculate this constant was Simpson's method due to its accuracy compared to, for example, the trapezoid method (see \cite{jorge2010}). Dividing by this constant we obtained proper probability densities.


\subsection{Comparison between continuous approximations}\label{section: comparison densities}

Since the allele frequency is obtained from a discrete process, and then compared to distributions that are continuous by definition, we need to either discretize the theoretical densities or approximate the empirical probability function of the simulated allele frequencies by a continuous density. Here, we chose to estimate the density of the allele frequencies through an adaptive method proposed by \citet{Bertin2014}.  From now on, we will denote the Adaptive Density Estimator as ADE.


The Beta distribution used for ADE has a bounded domain in the interval $[0,1]$, so it is appropriate for fitting allele frequencies. On one hand, the parameters can be estimated from its mean and variance, which have evolutionary sense. On the other hand, the Beta density estimator proposed by \citet{Bertin2011} is characterized through an unique tunning parameter, $b$, which determines the shape of the distribution, even being able to capture high concentrations at values close to 0 or 1. Then, the authors deal with an uniparametric Beta distribution: given a sample $X_1, \dots, X_n$, \citet{Bertin2011} defines the Beta kernel estimator as
\begin{equation*}
    \hat{f}_b(t) = \dfrac{1}{n} \sum_{k=1}^{n} K_{t,b} \left( X_{t_{k}} \right) , \quad t \in [0,1],
\end{equation*}
where
\begin{equation*}
    K_{t,b} (x) = \dfrac{x^{t/b} (1-x)^{(1-t)/b}}{B \left( \frac{t}{b} + 1, \frac{1-t}{b} + 1 \right)} \mathbf{1}_{(0,1)}(x).
\end{equation*}

To estimate the value of the $b$ parameter, we used an adaptive method for kernels from the Beta distribution, from \citet{Bertin2014}, based on the Lepski approach, which is a usual procedure for obtaining adaptive estimators. The primary objective of the cited work is to identify a trade-off between bias and variance (bandwidth problem in density estimation). 







The difference of this estimator and the proposed densities is that it is model-free, allowing to have the best continuous estimation of the density, but which is not very useful when making inference, for example, for testing or selection. This estimator allows us to have a continuous density and then compute distances between the model-free and the derived or most used densities for modelling the allele frequencies. Given the symmetry of the distributions with respect to the initial allele frequency, $x_0$, we demonstrate the results for the initial values $x_0=0.1$, $0.3$, and $0.5$. A remarkable fact is that the adaptive estimator obtained, fits well for all values initial values considered, and for the different generations.

\subsubsection*{Hellinger distance}


To compare the fit of the theoretical approximations to the ADE we will use the Hellinger distance between two distributions. Given two probability density functions $p$ and $q$, we define the Hellinger distance between the two as
\begin{equation}\label{eq: hell-cont}
    H^2(p,q) = \frac{1}{2} \int_{-\infty}^{+\infty} \left (\sqrt{p(x)} - \sqrt{q(x)} \right)^2 \, d x.
\end{equation}
An equivalent definition can be made for two probability mass functions $p$ and $q$ defined on $\mathcal{X} = \{x_i\}_{i \in \N}$; in this case the Hellinger distance has the following form
\begin{equation}
    H^{2} (p,q) = \dfrac{1}{2} \sum_{i=1}^{\infty} \left(\sqrt{p(x_i)} - \sqrt{q(x_i)}  \right)^{2}.
\end{equation}

\cite{Tataru2015} used the discrete version of the Hellinger distance to quantify and evaluate the distance between the distributions, discretizing the continuous approximation. We considered that, as these distributions concentrate probability on 0 and 1, and the problem is how we can approximate the densities with a continuous one, the discrete approach could be unfair with some distributions. Thus, we decided to compute the Hellinger distance for continuous distributions, as in \eqref{eq: hell-cont}, because discretizing the continuous distributions could potentially hide the problems that pose the atoms in the W-F model.

\begin{remark}
The following repository, \url{https://github.com/gerardo-martinez-j/wf-density}, contains the codes used to simulate this study 
\end{remark}

\subsection{Results} \label{section: results}

To understand how the different approximations fit the data, given the initial allele frequencies ($x_0$) and times ($t$), in Figure \ref{comparisons} we show the histograms of the allele frequencies obtained after 100 simulations of the W-F process, the densities of the Asymptotic Expansion (AE) and the Gaussian approximation (GaussA) the density obtained through the adaptive method (ADE). 
The generations included in the figure are $n = 0.100 \times 2N, 0.25\times 2N$ and $ 0.45\times 2N$, for the initial allele frequencies $x_0 = 0.1$, $x_0 = 0.3$, $x_0 = 0.5$.

We can see that the ADE fits well with the histograms for most of the $t$s and $x_0$s (see Discussion for an explanation about the goodness of fit and the variance of the estimator). From the proposed densities, we see that the Asymptotic Expansion captures better the shape of the histogram (getting closer to the ADE fit) than the Gaussian Approximation, especially for the higher number of generations and extreme $x_0$s. The latter, despite showing closeness to ADE, does not capture adequately the effect of accumulation at 0 or 1 as the value of $t/2N$ increases(Figure \ref{img: heatmap-hellinger}).


\begin{figure}[H]
      \centering
 \includegraphics[width=1.2\textwidth]{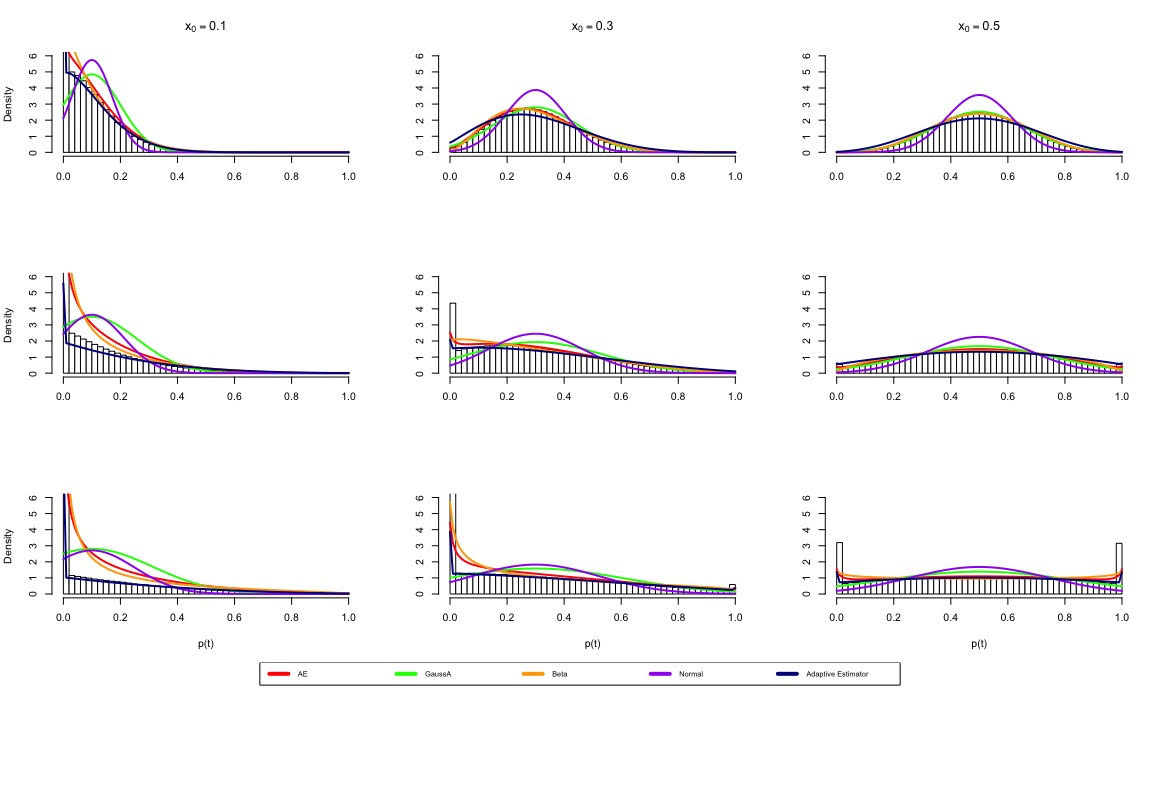}   
         

\caption{Comparison of the distributions proposed with empirical histogram, for a particular generation and different initial values of $x_{0}$.}
\label{comparisons}
\end{figure}

Figures \ref{img: heatmap-hellinger} and \ref{heatmap-L2} (see Appendix) show the comparisons between Hellinger distance and $L^2$ norm, calculated between the density estimate for the allele frequency, obtained through the adaptive method and the proposed densities: Gaussian approximation (GaussA) and Asymptotic expansion (AE); and the most commonly used densities: Beta and Gaussian distribution. 

\begin{figure}[H]
\centering
\includegraphics[scale = 0.65]{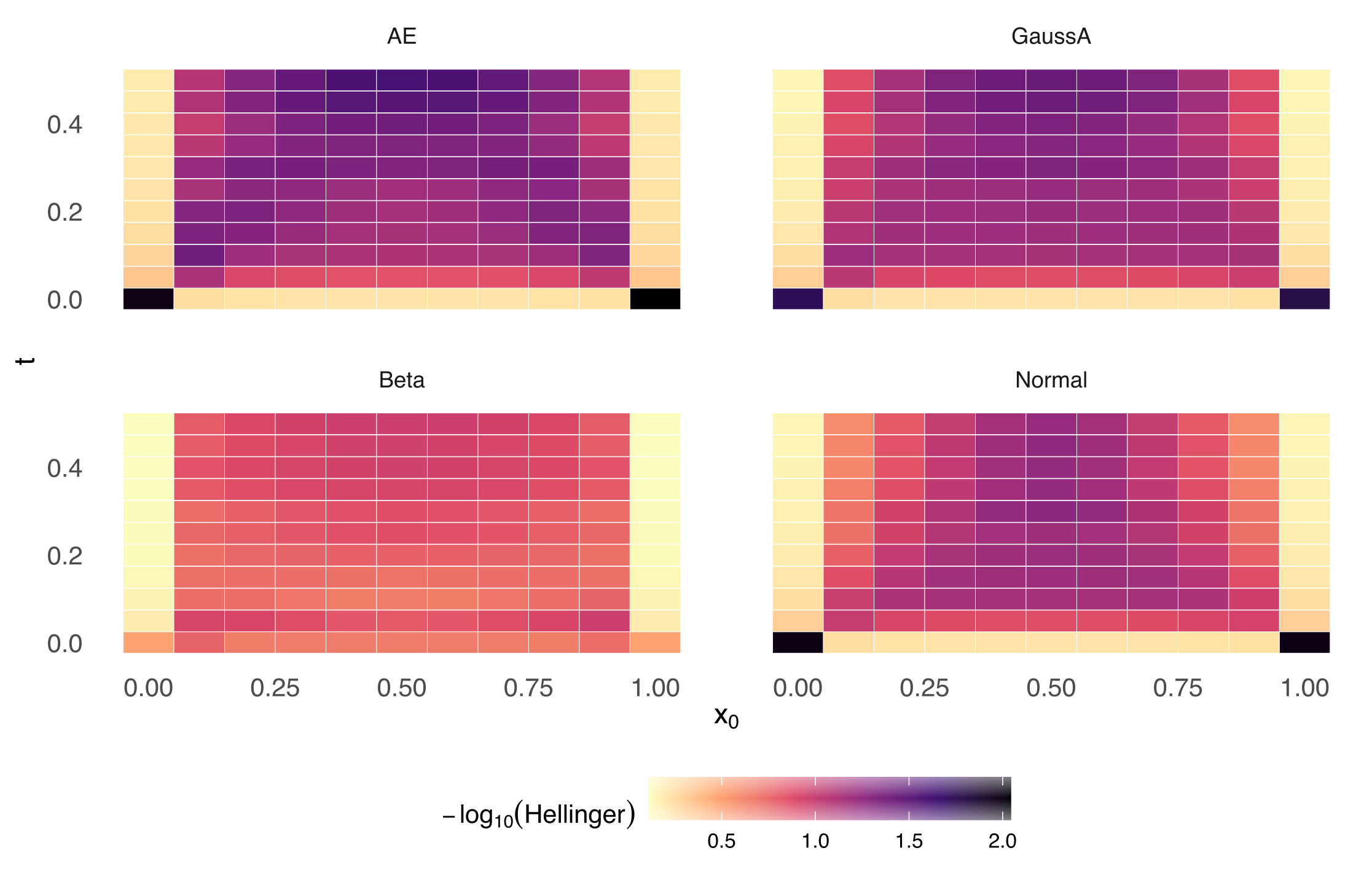}
\caption{\footnotesize Heatmaps of $-\log_{10}(H(p,q))$, with $p(x)$ being the ADE and $q(x)$ being the densities defined in section \ref{section: comparison densities}.}
\label{img: heatmap-hellinger}
\end{figure}

The values of the Hellinger distance and $L^2$ norm are smaller than Gaussian and Beta for both, Asymptotic expansion and Gaussian approximation. This relationship is illustrated by the color scale, where darker colors correspond to lower values.. Moreover, it is possible to notice that the commonly used distributions (Beta and Gaussian distributions) do not capture good fits for values close to 0 or 1 when the value of $t/2N$ is large. Whereas the distributions proposed in our work fit better to this feature of the allele frequency distribution. However, for values of $t/2N$ between 0.1 and 0.9, approximately, all the studied distributions fit well. That is when not much concentration is observed at values close to 0 or 1.

\subsection{Discussion} \label{section: discussion}
In this work we derive two approximated distributions for the allele frequency of a W-F process of a fixed time $t=\frac{n}{2N}$, where $N$ is the constant population size and $n$ the number of generations. 
After evaluating the approximations through simulations, we conclude that the AE captures better the accumulation of allele frequencies closer to 0 and 1, that can be see in Figure \ref{img: heatmap-hellinger}.

The Hellinger distance between the density obtained by the ADE and the densities studied (AE, GaussA, Beta, and Normal) shows that, in general, the densities fit well for the different values of $t$ and $x_0$, except for the generation $t=0.001$ and for values of $x_0$ too close from zero and one. The distances of AE and GaussA are clearly smaller  than the distance to the Beta approximation, but also than to the Normal density. None of the distributions can capture properly the fixation at 0 or 1.

The DAF estimator, obtained through the adaptive method for Beta Kernels, allowed us to have a continuous density of a process that, in essence, is discrete and allowed us to compare against distributions that are continuous. This method is based on a trade-off between bias and variance, thus controlling the bandwidth-problem, usual in density estimation problems. As it is possible to notice in Figure \ref{comparisons}, the histogram (DAF) is not exactly the same as the density given by the adaptive method (blue line), particularly for the case when $x_0=0.5$; this is because when $x_0=0.5$, we are in the case with the highest variance of the DAF. The Hellinger distance presented in the work of \cite{Tataru2015}, was calculated in its discrete version, i.e., discretizing densities that are continuous by definition; while in our work, we calculated the continuous version of it. This leads to different results, particularly close to 0 and 1. Here, we wanted to use the continuous versions, as we considered that discretizing the densities could hide the bad fit for values close to 0 and 1.

The analytical expresion for the transition density, \eqref{eq: density of x 1}, provides the opportunity to set up problems related to parametric inference within the W-F model. These include estimating the initial allele frequency, $\hat{x}_0$, which is used to infer demographic history. However, to carry out this procedure, it is necessary to have multiple populations and, therefore, the respective joint density. Our work, however, only considers one population, so the extension to multiple populations through the explored method represents a new problem. Authors such as \cite{gutenkunst2009} propose to approach the problem by calculating approximations, involving only up to three simultaneous populations, but without having an exact analytical formula that allows a deeper study of the properties of the estimated parameters.





\subsection*{Acknowledgement}
T.Roa thanks to Proyecto FONDECYT Postdoctorado 3220043, 19-MATH-06 SARC,  FONDECYT 1171335 projects and Project Code 14ENI2-26865. T. Roa and R. León thanks to CONICYT-ANID: MATHAMSUD FANTASTIC 20-MATH-05 project. 



\newpage

\newpage
\section{Appendix} \label{apend}

\subsection{Additional Figures}

\begin{figure}[H]
\centering
    \includegraphics[scale=0.65]{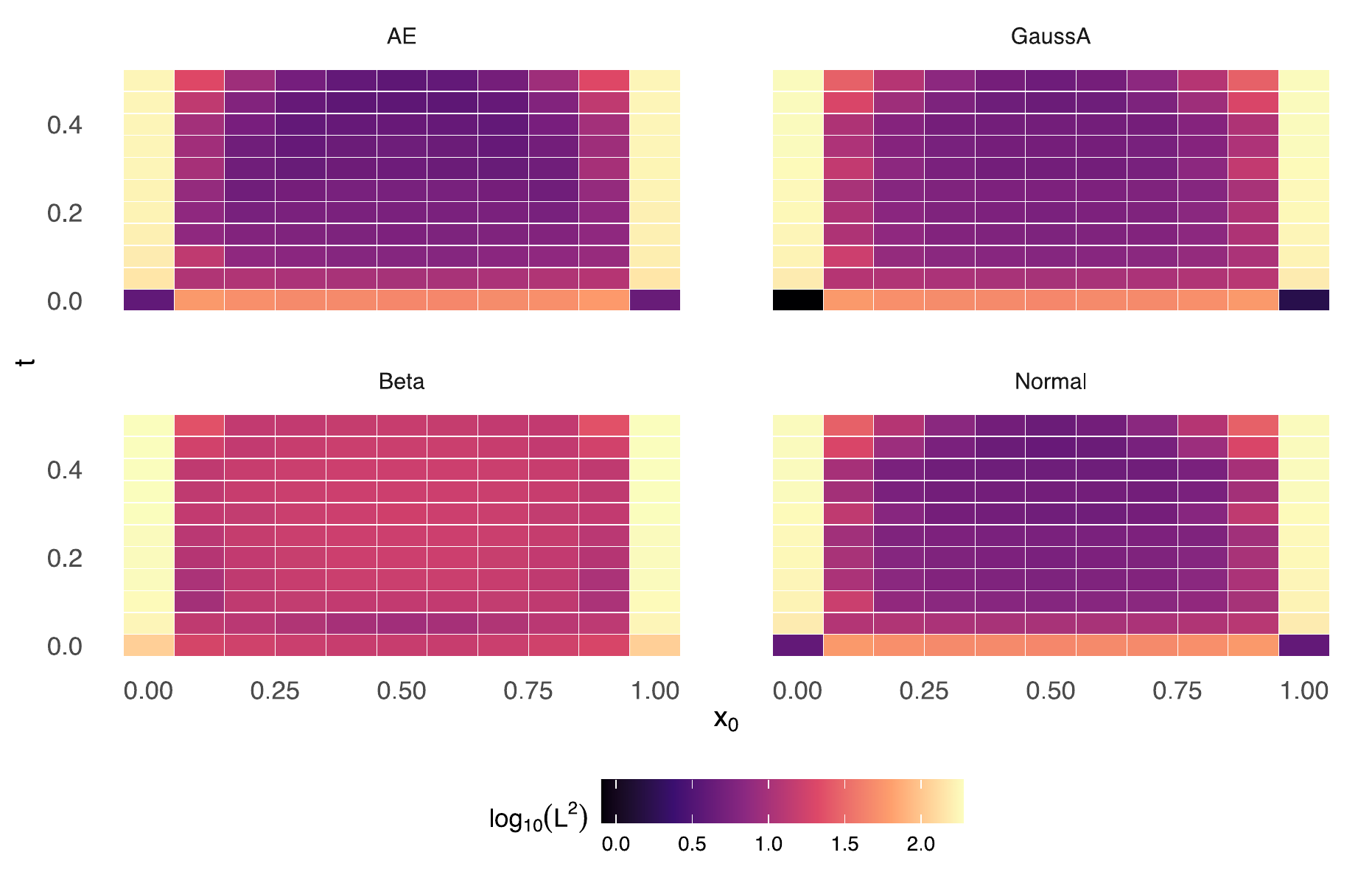}
\caption{Heatmaps of $\log_{10}(L^2)$ for Asymptotic expansion (AE), Gaussian approximation (GaussA), Beta and Normal distribution.}
\label{heatmap-L2}
\end{figure}

\subsection{Mutation} \label{OSC1}
In this section we will prove the theorem \ref{theo: density W-F mutation} for the case where there is no selection. This model has the local variance term equal to $\sigma(x)=\sqrt{x(1-x)}$ and the drift term is equal to $\mu(x)=-\beta_1 x+\beta_2(1-x),$ where $\beta_1$ and $\beta_2$ are the scaled mutation rates.

Again, we will make use of the transformed process $Y(t)=F(X(t))$, with $F$ defined in (\ref{eq: definition of g}). Using Itô's formula, we get
\begin{equation*}
    dY(t)=dW(t)-\Big(\frac{1}{4}\frac{1-2X(t)}{\sqrt{(X(t)(1-X(t))}}+\frac{(\beta_1 X(t)-\beta_2(1-X(t))}{\sqrt{(X(t)(1-X(t))}}\Big)dt,
\end{equation*}
that can be written using the explicit formula for $F^{-1}$ as
\begin{equation*}
    dY(t)=dW(t)+\left [\left (\frac{1}{4}-\beta_1 \right )\tan \left (\frac{Y(t)}{2} \right )-\left (\frac{1}{4}-\beta_2 \right )\cot \left (\frac{Y(t)}{2} \right)\right ] dt.
\end{equation*}

By defining the new function drift as
\begin{equation*}
    \tilde \mu(y)=\left (\frac{1}{4}-\beta_1 \right)\tan \left (\frac{y}{2} \right)- \left (\frac{1}{4}-\beta_2 \right)\cot \left (\frac{y}{2} \right).
\end{equation*}
With the same notation of Section \ref{section: AE} we get as in (\ref{eq: AE1})
\begin{equation}\label{primera1}
p^Y_t(y_0,y)=q_t(y-y_0)\exp \left (\int_{y_0}^y \tilde{\mu}(u)du \right )(1+O(t)).
\end{equation}
In this form, we obtain
\begin{IEEEeqnarray*}{rCl}
    M(y)-M(y_0) & = &\int_{y_0}^{y} \left[ \left(\frac{1}{4}-\beta_1 \right) \tan \left(\frac{u}{2}\right)- \left(\frac{1}{4}-\beta_2\right) \cot \left(\frac{u}{2} \right)\right]du \\
    & = & 2\int_{\frac{y_0}2}^{\frac y2} \left[ \left(\frac{1}{4}-\beta_1 \right) \tan(u)- \left(\frac{1}{4}-\beta_2 \right)\cot(u) \right]du \\
    & = &2\Big(- \left(\frac{1}{4}-\beta_1 \right)\log|\cos u|- \left(\frac{1}{4}-\beta_2 \right)\log|\sin u|\Big)\Big|_{\frac{y_0}2}^{\frac y2} \\
    & = & -\log\Big(|\cos u|^{2(\frac{1}{4}-\beta_1)}|\sin u|^{2(\frac{1}{4}-\beta_2)}\Big)\Big|_{\frac{y_0}2}^{\frac y2} \\
    & = &-\log\Big( \left(\cos^2 \left(\frac{y}2 \right) \right)^{(\frac{1}{4}-\beta_1)} \left(\sin^2 \left(\frac{y}2\right) \right)^{(\frac{1}{4}-\beta_2)}\Big) \\
    & &+\log\Big( \left(\cos^2 \left(\frac{y_0}2\right) \right)^{(\frac{1}{4}-\beta_1)} \left(\sin^2 \left(\frac{y_0}2 \right) \right)^{(\frac{1}{4}-\beta_2)}\Big).
\end{IEEEeqnarray*}
Thus taking exponential we have
\begin{equation*}
    \exp \left [M(F(x))-M(F(x_0)) \right ]= \left (\frac{1-x_0}{1-x} \right )^{(\frac{1}{4}-\beta_1)} \left (\frac{x_0}{x} \right )^{(\frac{1}{4}-\beta_2)}.
\end{equation*}

As we have done to obtain the equation \eqref{eq: AE1}, it results that the following asymptotic expansion holds 
\begin{equation*}
    p^X_{t}(x_0,x)=\frac1{\sqrt{2\pi t}} \frac{(1-x_0)^{(\frac{1}{4}-\beta_1)}x_0^{(\frac{1}{4}-\beta_2)}}{(1-x)^{\frac{3}{4}-\beta_1}x^{\frac{3}{4}-\beta_2}}e^{-\frac{1}{2t}(\int_{x_0}^{x}\sigma^{-1}(u)du)^2}(1+O(t)).
\end{equation*}

\subsection{Selection} \label{OSC2}
In this section, we consider the model including selection but no mutation.  In such a case, we keep the same diffusion coefficient from \eqref{eq: WF-selection}, but now the drift is defined as in \citet{Ewens2004}, Chapter 5
\begin{equation*}
    h(x)=\alpha x(1-x)\Big(x+h(1-2x)\Big).
\end{equation*}
By (\ref{eq: transformed process 1}), we get
\begin{IEEEeqnarray*}{rCl}
    dY(t)&=&dW(t)+\Big(\frac{(4h-1)}{4}\frac{1-2X(t)}{\sqrt{X(t)(1-X(t)}}+\alpha\sqrt{\frac{X(t)}{(1-X(t))}}\Big)dt\\
    &=&dW(t)+\Big(\frac{(4h-1)}{2}\cot (Y(t))+\alpha \tan \left(\frac{Y(t)}2\right)\Big)dt.
\end{IEEEeqnarray*}
Thus
\begin{equation*}
    M(y)-M(y_0)=\frac{(4h-1)}{2}\int_{y_0}^{y}\cot u\,du+\alpha\int_{y_0}^y\tan\frac u2\,du.
\end{equation*}

Proceeding as in the previous section, we get
\begin{equation*}
    M(F(x))-M(F(x_0))=\ln \left( \frac{x_0(1-x_0)}{x(1-x)} \right)^{(\frac14-h)}+\ln \left(\frac{1-x_0}{1-x} \right)^{\alpha}.
\end{equation*}
And
\begin{equation*}
    e^{(M(F(x))-M(F(x_0)))}= \left( \frac{x_0(1-x_0)}{x(1-x)} \right)^{(\frac14-h)} \left(\frac{1-x_0}{1-x} \right)^{\alpha}.
\end{equation*}

The rest  is plain.

\end{document}